\documentclass[11pt]{article} \usepackage[english]{babel}
\usepackage{amstext,amsfonts,amsmath,amssymb} \usepackage{amsthm}
\usepackage{color}
\parskip=\medskipamount
\theoremstyle{definition}
\newtheorem{definition}{Definition}
\newtheorem{example}{Example} 

\def\pd#1#2{\frac{\partial#1}{\partial#2}}
\newtheorem{theorem}{Theorem}
\newtheorem{proposition}[theorem]{Proposition}
\newtheorem{lemma}[theorem]{Lemma}
\def\CR{{\mathbb R}}
\def\FX{{\mathfrak X}}
\def\D{{\mathcal D}}

\def\x{\times}
\font\frak=eufm10 scaled\magstep1
\def\goth #1{\hbox{{\frak #1}}}
\def\core#1{{{#1}_*}}
\def\janusz{}
\begin{document}

\title{Geometry of Lie integrability by quadratures}

\author{
J.F. Cari\~nena$^{\dagger\,a)}$,
F. Falceto$^{\ddagger\,b)}$,
J. Grabowski$^{\diamond\,c)}$, and
M.F. Ra\~nada$^{\dagger\,d)}$  \\ \\
${}^{\dagger}$
   {\it Departamento de F\'{\i}sica Te\'orica and IUMA, Facultad de Ciencias} \\
   {\it Universidad de Zaragoza, 50009 Zaragoza, Spain}  \\
${}^{\ddagger}$
   {\it Departamento de F\'{\i}sica Te\'orica and BIFI, Facultad de Ciencias} \\
   {\it Universidad de Zaragoza, 50009 Zaragoza, Spain}  \\
${}^\diamond$
  {\it Polish Academy of Sciences, Institute of Mathematics,   }
  \\{\it  Sniadeckich 8, PO Box 21,  00-656 Warsaw, Poland }
 }
\date{\today}
\maketitle

\begin{quote}
\end{quote}

\begin{abstract}
In this paper we extend the Lie theory of integration in two different ways.
First we consider a finite dimensional Lie algebra of vector fields and
discuss the most general conditions under which
the integral curves of one of the fields can be obtained
by quadratures in a prescribed way. It turns out that the conditions
can be expressed in a purely algebraic way.
In a second step we generalize the construction to the case in which
we substitute the Lie algebra of vector fields by {\janusz  a module (generalized distribution). We obtain much larger class of integrable systems replacing standard
concepts of solvable (or nilpotent) Lie algebra with distributional solvability (nilpotency).}
\end{abstract}

\begin{quote}

{\it MSC Classification:}
{\enskip}37J35, {\enskip}34A34, {\enskip}34C15, {\enskip}70H06
\end{quote}

{\vfill}

\footnoterule
{\noindent\small
$^{a)}${\it E-mail address:} {jfc@unizar.es } \\
$^{b)}${\it E-mail address:} {falceto@unizar.es }\\
 $^{b)}${\it E-mail address:} {jagrab@impan.pl }\\
$^{c)}${\it E-mail address:} {mfran@unizar.es}    }
\newpage

 \tableofcontents
\section{Introduction}

Integrability of a given system of differential equations is a recurrent subject of very much interest and it has been an active field of research along the last years.
The meaning of integrability, however, is not well-defined and has a different sense
within each theory, and it is only rigourously defined  in each specific field.
Of course integrability means that you can find the general  solution in an algorithmic way,  but for instance you can restrict yourself to search for solutions of
 a previously selected class of functions, polynomial functions, rational functions, etc. The existence of additional structures,
for instance compatible symplectic structures,
may be useful. The possible dependence of integrability of a lucky choice of coordinates is also a relevant point.

The objective of this work is to reanalyze the classical problem of integrability by quadratures, without any recourse to the existence of additional compatible structures,
but using modern tools of algebra and geometry.
The first relevant result  is due to Lie and we aim to slightly generalize the result of his work.  Partial integrability is related to the existence of some first integrals
and infinitesimal symmetries, but in the generic case we may not have enough number of them for an effective
finding of the complete solution and for most part of physical systems we will be unable to answer about integrability or non-integrability of the system.
This gives even more importance to results about characterization of particular cases in which such a question can be answered.
This has motivated reinvention of many integration techniques that had been previously introduced by distinguished mathematicians of past centuries.
We will fix our attention on Lie approach to the problem which was based on the use of Lie algebras of symmetry vector fields, and more in particular solvable Lie algebras \cite{SL81}.
For a recent description of other related integrability approaches see e.g. \cite{MCSL}.

{\janusz  After proving a generalization of the classical result by Lie, motivated by some natural examples which do not fit into Lie's scheme, we develop new ideas in which Lie algebras
of vector fields are replaced by certain modules of vector fields (distributions). The corresponding concept of \emph{distributional integrability} allows for much larger class of examples and potential applications. Our approach and the corresponding results are, up to our knowledge, novel and  original.}

The paper is organized as follows: Section 2 is devoted to introduce notation and establishing the relation of integrability by quadratures with
 the standard Arnold-Liouville integrability. In particular we recall the classical theorem by Lie on integrability by quadratures whose proof, for the simplest case $n=2$, is given.
 In section 3 we recall some concepts of cohomology needed to analyze the existence of solutions for a system of differential equations. A lemma establishing
 in cohomological terms necessary conditions for the existence of solution of a first order system of differential equations is given.
In section  4 we introduce an iterative process for solving a system of first order differential equations, expressed in geometrical terms as a vector field $\Gamma$ on a manifold $M$.
The procedure consists in constructing by quadratures a  sequence of nested Lie subalgebras $L_{\Gamma,k}$ of a $\dim(M)$-dimensional Lie algebra of vector fields $L$, such that
any of them contain the dynamical vector field $\Gamma$. If in some step the resulting Lie algebra is Abelian,
 we can obtain with one more quadrature the general solution. Some interesting algebraic properties are studied
in Section 5, and in particular
we prove the in order for the process, outlined above, to work the
 Lie algebra must be solvable. Conversely, it is proved that when a solvable Lie algebra $L$ contains an Abelian ideal $A$,
then $(M,L,\Gamma)$ is Lie integrable, i.e. the previous algorithm works,
for any $\Gamma\in A$.
A particular case is the theorem by Lie discussed in section 2.
We also prove that  if the Lie algebra $L$ is nilpotent,
then $(M,L,\Gamma)$ is Lie integrable for any $\Gamma\in L$.

An interesting example which has been recently studied from a Hamiltonian viewpoint is reanalyzed in section 6 without any recourse
to the symplectic structure of the phase space, but
 focusing our attention on the Lie algebra structure of the symmetries. Last section is devoted to extending the previous results to
the more general situation in which, instead of having a Lie algebra $L$ of vector fields,
we have a vector space $V$ such that its elements do not close
a finite dimensional real Lie algebra, but they generate a  general integrable distribution of vector fields.
{\janusz  We develop new geometric approach to integrability based
on ceratin algebraic properties of distributions. The introduced concept of \emph{distributional integrabilty} and the corresponding version of Lie's theorem provides a large new class
of systems integrable by quadratures. We also prove in this context results similar to that of Section 5, using original ideas of \emph{distributional solvabilty} and \emph{distributional nilpotency}.}

\section{Integrability by quadratures}
 Recall that an autonomous system of differential equations,
\begin{equation}
\dot x^i=  f^i(x^1,\ldots,x^n)\ ,\qquad  i=1,\ldots,n, \label{autsyst}
\end{equation}
 is geometrically interpreted in terms of a vector field $\Gamma$ in a
$n$-dimensional manifold $M$ with
a local expression
\begin{equation}
\Gamma=\sum_{i=1}^nf^i(x^1,\ldots,x^n)\partial_i\ ,\qquad \partial_i\equiv\pd{}{x^i}\ . \label{leavf}
\end{equation}
The integral curves of $\Gamma$ are the solutions of (\ref{autsyst}). Integrating the system amounts
to determine its general solution. In particular, integrability by quadratures means that you can determine the solutions by means
of a finite number of algebraic operations and integrals of known functions.

The two main techniques for solving the system
are  the determination  of first-integrals and the search for infinitesimal symmetries of the vector field. The set of first-integrals provides  us a
foliation such that the vector field is tangent to the leaves and reducing in this way the problem to a family of lower dimensional ones, one
 in each leaf, while the knowledge of symmetries of the vector field,   suggests us to use adapted coordinates,  the system  decoupling then
  into  lower dimensional subsystems.

   More specifically, if $F_1,\ldots,F_r$, are functionally independent, i.e. such that $dF_1\wedge\cdots\wedge dF_r\ne 0$, first-integrals, then
for a given a vector field  $\Gamma\in   \mathfrak{X}(M)$, we can consider the foliation whose leaves are the level sets of the  function of rank $r$,
  $(F_1,\ldots,F_r):M\to \mathbb{R}^r$, and as $\Gamma$ is tangent to each leave, the problem is reduced
 to that of the vector fields
 $\widetilde \Gamma_c$  defined in each   $n-r$ dimensional leave
 $M_c=\mathbf{F}^{-1}(c)$, $c\in\mathbb{R}^r$. Of course, the best situation is when $r=n-1$ because then the leaves
to be considered are one-dimensional, giving us the solutions to the problem, up to a reparametrization.

The other way of reducing the problem is based on the knowledge of infinitesimal (or one-parameter subgroups of)
symmetries, i.e.
vector fields  $X$ such
that  $[X,\Gamma]=0$.  The result of the Straightening out Theorem \cite{crampinbook} asserts the existence of adapted coordinates $(y^1,\ldots,y^n)$ in a neighbourhood of a point where $X$ is different from zero, i.e. such that
$$X=\pd{}{y^n}\ ,$$
and its integral curves are obtained by solving a subsystem involving only the other $n-1$ coordinates. Note however that the new coordinates $y^1,\ldots,y^{n-1}$, are constants of motion
 and therefore we cannot find easily such coordinates in a general case.
Moreover,  the information provided by
 two different symmetry vector fields cannot be used simultaneously unless they commute.

 It is clear that the if we use such rectifying coordinates for $\Gamma$ the integration is immediate, the solution being
 $$y^k(t)=y^k_0, \quad k=1,\ldots ,n-1,\qquad y^n(t)=y^n_0+t.
 $$
 This proves that the concept of integrability by quadratures depends on the choice of initial coordinates, because using these adapted coordinates the system is always integrable by quadratures.

 Both constants of motion and infinitesimal symmetries can be used simultaneously if some compatibility conditions are satisfied.
We can say that a system admitting $r<n-1$ functionally independent constants of motion is integrable when we know furthermore $s$   infinitesimal symmetries $X_1,\ldots,X_s$, with $r+s=n$  such that
$$[X_a,X_b]=0, \quad a,b=1,\ldots, s,\qquad
\textrm{and}\qquad X_aF_\alpha=0, \quad \forall a=1,\ldots,s,\ \alpha=1,\ldots r.
$$
The constants of motion determine a $n-r$ foliation  and the former condition
means that the restriction of vector fields $X_a$ to the leaves are tangent to such leaves.

Sometimes we have additional geometric structures that are compatible with the dynamics
 \cite{Pryk99}.
For instance a $2m$-dimensional manifold $M$ is endowed with a symplectic structure $\omega$. Such 2-form relates, by contraction,  in a one-to-one way vector fields and 1-forms, and vector fields $X_F$ associated with exact 1-forms $dF$ are said to be Hamiltonian vector fields. Compatible here means that the dynamical vector field itself is a Hamiltonian vector field $X_H$. The particularly interesting case of
Arnold--Liouville  definition of (Abelian) complete integrability \cite{A78,JL53,FM78,K83} is a particular case where $r=m$,
the vector fields are $X_a=X_{F_a}$ and, for instance, $F_1=H$.
The regular Poisson bracket defined by $\omega$, i.e. $\{F_1,F_2\}=X_{F_2}F_1$, allows us to express the above tangency conditions as
$$X_{F_b}F_a=\{F_a,F_b\}=0,  \quad a,b=1,\ldots, m,
$$
i.e. the $m$ functions are constants of motion in involution and the corresponding Hamiltonian vector fields commute.

Our aim in this paper is to study integrability in absence of additional compatible structures, the main tool being the symmetries of the given vector field very much in the approach started by Lie.
We will see that if the given vector field is part of an
appropriate Lie algebra of vector fields, then it is integrable by quadratures in any chart.

Given a  vector field $\Gamma\in   \mathfrak{X}(M)$ in a differentiable
 manifold $M$, the set of strict symmetries of
$\Gamma$ is a linear space  $   \mathfrak{X}_\Gamma(M)=\{X\in   \mathfrak{X}(M)\mid [X,\Gamma]=0\}$.
Obviously, $\Gamma\in  \mathfrak{X}_\Gamma(M)$. The flow of a vector field
 $X\in   \mathfrak{X}_\Gamma$ preserves the
set of integral  curves of the dynamical vector field $\Gamma$.
There are  vector fields generating flows preserving the set of integral
curves of $\Gamma$ up to a reparametrization, those of the set
$   \mathfrak{X}^\Gamma(M)=\{X\in   \mathfrak{X}(M)\mid [X,\Gamma]=f_X\, \Gamma\}$,
where $f_X\in C^\infty (M)$. It is also a real
linear space containing $  \mathfrak{X}_\Gamma(M)$. Actually, vector fields in $  \mathfrak{X}^\Gamma(M)$
preserve the one-dimensional distribution generated by $\Gamma$.

Notice that $  \mathfrak{X}_\Gamma(M)$ is a Lie subalgebra of the real Lie algebra $  \mathfrak{X}^\Gamma(M)$, as a consequence of Jacobi identity, but is not an ideal.

The problem of integrability by quadratures depends on the determination by quadratures of the necessary first-integrals and on finding adapted coordinates, or,  in another words, finding a sufficient number of tensor invariants \cite{Koz13,B92}.
 The first result is due to Lie
who established the following theorem:

\begin{theorem}\label{Lie-thm}
If $n$ vector fields $X_1$,\ldots,$X_n$, which are linearly independent {\janusz  at each point} of  an open set  $U\subset\mathbb{R}^n$, generate a solvable Lie algebra and  are  such that $[X_1,X_i]=\lambda_i\, X_1$ with $\lambda_i\in \mathbb{R} $,
then the differential equation  $\dot x_i=X_1\;x_i$ is solvable by quadratures in $U$.
\end{theorem}
\begin{proof}
We only prove here the simplest case $n=2$ (the higher dimensional one
is a particular case of our Theorem 3).
Then, the derived algebra is one-dimensional and therefore the Lie algebra is solvable. The differential equation can be integrated if
we are able to find a first integral $F$ for $X_1$, i.e. $X_1F=0$,  such that $dF\ne 0$ in $U$.  In this case we can implicitly define one variable, for instance $x_2$,
 in  terms of the other one by $F(x_1,\phi(x_1))=k$,  and the differential equation determining the integral curves of $X_1$ is  in separate variables, i.e. integrable by quadratures.

Let $X_1$ and $X_2$  be  two vector fields such that $[X_1,X_2]=\lambda_2\, X_1$.  note that since $n=2$   there exists a 1-form $\alpha_0$,
which is defined up to multiplication by a function, such that $i(X_1)\alpha_0=0$. Obviously as $X_2$ is linearly independent of $X_1$ at each point, $i(X_2)\alpha_0\ne 0$.
The 1-form $\alpha=(i(X_2)\alpha_0)^{-1}\alpha_0$ satisfies the condition $i(X_2)\alpha=1 $, by definition, together with $i(X_1)\alpha=0$,  and we can  see that $\alpha $ is then closed, because  $X_1$ and $X_2$ generate $\mathfrak{X}(\mathbb{R}^2)$ and
$$ d\alpha(X_1,X_2)=X_1\alpha(X_2) -X_2\alpha(X_1)+\alpha([X_1,X_2])= \alpha([X_1,X_2])=\lambda_2\, \alpha(X_1)=0.
$$
Therefore, there exists, at least locally, a function $F$ such that $\alpha=dF$. In other words, the function $(i(X_2)\alpha_0)^{-1}$ is an integrating function for $\alpha_0$.  The condition $i(X_1)\alpha=0$ means that $F$ is a first integral $F$ for $X_1$.

{\janusz  The (locally defined) function $F$ such that
$$F(x^1,x^2) =\int_{\gamma_{(x^1,x^2)}}\alpha,
$$
where $\gamma_{(x^1,x^2)}$ is any curve joining a reference point $(x^1_0,x^2_0)\in U$ with the point $(x^1,x^2) $, is the function we were looking for.}
\end{proof}

We do not give the proof for $n>2$ because a more general result will be proved later on.

\section{Local integration by quadratures of closed forms}

Let us first recall some useful notions of Lie algebra cohomology. Let  $\mathfrak{g}$  be a Lie algebra and  $\mathfrak{a}$ a $\mathfrak{g}$-module respectively. In other
words, $\mathfrak{a}$ is a linear space  that is the carrier space for a linear
representation $\Psi$ of  $\mathfrak{g}$, i.e. $\Psi \colon \mathfrak{g} \to \textrm{End\,} \mathfrak{a}$ satisfies
$$\Psi (a) \Psi (b)-\Psi (b) \Psi(a)=\Psi ([a,b]),\quad \forall a,b\in \mathfrak{g}.$$

By a  $k$-cochain  we mean a $k$-linear alternating
mapping from $\mathfrak{g}\x\cdots\x\mathfrak{g}$ ($k$~times) into $\mathfrak{a}$.  We denote by
$C^k(\mathfrak{g},\mathfrak{a})$
the space of $k$-cochains. For every $k\in\mathbb{N}$ we define
$\delta_k:C^k(\mathfrak{g},\mathfrak{a})\to
C^{k+1}(\mathfrak{g},\mathfrak{a})$ by \cite{CE48,CI88}
$$\begin{array}{rcl}
(\delta_k\alpha)(a_1,\dots,a_{k+1})
 &= &{\displaystyle\sum_{i=1}^{k+1} (-1)^{i+1} \Psi(a_i)
        \alpha(a_1,\dots,\widehat a_i,\dots,a_{k+1})+ } \\
&+& {\displaystyle\sum_{i<j} (-1)^{i+j}
\alpha([a_i,a_j],a_1,\dots,\widehat
                a_i,\dots,\widehat a_j,\dots,a_{k+1})},
\end{array}
$$
where $\widehat a_i$ denotes, as usual, that the element $a_i$ is omitted.

  The  linear maps $\delta _k$ can be shown to satisfy $\delta _{k+1}\circ
\delta _k=0$.
The linear operator $\delta$ on
$C(\mathfrak{g},\mathfrak{a}) = \bigoplus_{n=0}^\infty C^k(\mathfrak{g},\mathfrak{a})$ whose restriction to
each $C^k(\mathfrak{g},\mathfrak{a})$ is $\delta_k$, satisfies $\delta^2 = 0$.  We will then
denote
$$\begin{array}{rcl}
B^k(\mathfrak{g},\mathfrak{a}) &=& \{\alpha \in C^k(\mathfrak{g},\mathfrak{a}) \mid \exists\beta\in
C^{k-1}(\mathfrak{g},\mathfrak{a}) \text{ such that }\alpha = \delta \beta \}
  = \textrm{Image\,} \delta_{k-1},  \\
Z^k(\mathfrak{g},\mathfrak{a}) &= &\{\alpha\in C^k(\mathfrak{g},\mathfrak{a}) :\mid \delta\alpha = 0\} = \ker
\delta_k.
\end{array}
$$

The elements of $Z^k(\mathfrak{g},\mathfrak{a})$ are called $k$-cocycles, and those of
$B^k(\mathfrak{g},\mathfrak{a})$ are called $k$-coboundaries. Since $\delta^2 = 0$, we have
$B^k (\mathfrak{g},\mathfrak{a})\subset Z^k(\mathfrak{g},\mathfrak{a})$. The $k$-th cohomology group $H^k(\mathfrak{g},\mathfrak{a})$ is
defined as
$$
H^k(\mathfrak{g},\mathfrak{a}) := \frac{Z^k(\mathfrak{g},\mathfrak{a})}{B^k(\mathfrak{g},\mathfrak{a})} \,,
$$
and we will define $B^0(\mathfrak{g},\mathfrak{a})=0$, by convention.

An  interesting example   is that of $\mathfrak{g}$ being a finite-dimensional Lie subalgebra of $\mathfrak{X}(M)$ (the space of vector fields on a manifold $M$), $\mathfrak{a}=\bigwedge^p(M)$ {\janusz  (the space of $p$-forms on $M$)}, and the action given by $\Psi(X)\zeta=\mathcal{L}_{X}\zeta$. For instance the case $p=0$ has been used in
\cite{COW93} in the study of weakly invariant differential equations and $p=1,2,$ play an interesting role in mechanics (see e.g. \cite{CI88}). We will use next the particular case $p=0$. In this case the elements of $Z^1({\goth g},\bigwedge^0(M)=C^\infty(M))$ are linear maps
 $h:{\goth g}\to C^\infty(M)$ satisfying
 $$
 \mathcal{L}_{X} h(Y) - \mathcal{L}_{Y} h(X) = h([X,Y])\ ,\qquad X,Y\in \mathfrak{X}(M),
 $$
and those of $B^1({\goth g},C^\infty(M))$ are those $h$ for which
 $\exists g\in C^\infty(M)$ with
 $$h(X) = \mathcal{L}_{X}g\ .$$

At this point we first present the following preliminary lemma.

\begin{lemma}\label{l1}
Let  $\{X_1,\ldots,X_n\}$ be a set of $n$ vector fields that are linearly independent {\janusz  at each point} of a  $n$-dimensional manifold $M$. Then:

1) {\janusz  The necessary and sufficient condition for the system of equations for $f\in C^\infty(M)$
\begin{equation}\label{system}
  X_i f = h_i, \qquad h_i\in C^\infty(M) ,\quad i=1,\dots,n,
\end{equation}
to have a solution is that the 1-form $\alpha\in \bigwedge ^1(M)$ such that
$\alpha(X_i)=h_i$ be  an exact 1-form}.

2)  If the previous $n$ vector fields generate a $n$-dimensional real  Lie algebra $\mathfrak{g}$, i.e. there exist real numbers $c_{ij}\,^k$ such that $[X_i,X_j]=\sum_k c_{ij}\,^k\, X_k$, then
 the necessary condition for the system of equations to have a solution is that the
 $\mathbb{R}$-linear function
 $h:\mathfrak{g}\to C^\infty(M)$  defined by $h(X_i)=h_i$ is a cochain  that is a cocycle.

\end{lemma}

\begin{proof} 1) 
{\janusz  A function $f$ is a solution of (\ref{system}) if and only if $\alpha=d f$.

}
2) Consider $\mathfrak{a}=C^\infty(M)$ and  the cochain determined by the linear map $h$. Now the necessary condition for the existence of the solution is written as:
$$
X_i(h_j)-X_j(h_i)-\sum_kc_{ij}\,^k\, h_k=(\delta_1 h)(X_i,X_j)=0
$$
which is just the 1-cocycle condition.
\end{proof}

Most properties of differential equations are of a local character, and closed forms are
locally exact, therefore we can restrict ourselves to appropriate open subsets $U$ of $M$, i.e. open submanifolds, where every closed 1-form is exact, for instance assuming
$U$ to be arc-wise connected and simply connected. Then if $\alpha$ is closed it is locally exact, $\alpha=df$ in a certain open $U$, $f\in C^\infty(U)$,  and the solution of the system can be found  by one quadrature:
the solution  function $f$ is  given by   the quadrature
\begin{equation}\label{solution}
  f(x)=\int_{\gamma_x}\alpha,
\end{equation}
where $\gamma_x$ is any path joining some  reference point $x_0\in U$ with $x\in U$.
We also remark that
 $\alpha$ is exact, $\alpha=df$, if
and only if  $\alpha(X_i)=df(X_i)=X_if=h_i$, i.e. $h$ is a coboundary,  $h=\delta f$.

 In the particular case of the functions $h_i$ appearing in the system being constant
the condition for the existence of local solution reduces to $\alpha([X,Y])=0$, for
each pair of elements, $X$ and $Y$ in  $\mathfrak{g}$, i.e. $\alpha$ vanishes on the derived Lie
algebra $\mathfrak{g}'=[\mathfrak{g},\mathfrak{g}]$.

\section{A generalization of the Lie theory of integration.}


Let us consider a family of $n$ vector fields,  $X_1,\dots,X_n$, defined on
a $n$-dimensional manifold $M$.
We  assume that they close a Lie algebra $L$ over the real
numbers
$$
 [X_i,X_j] = \sum_kc_{ij}\,^k \,X_k \,, {\quad} i,j,k = 1,\dots,n,
$$
and that in addition  they span a basis of $T_xM$ at  every point $x\in M$.
We pick up an element in the family, let us assume it is $X_1$,
that is going to be the dynamical vector field. In order to emphasize its
special r\^ole we will often denote it by $\Gamma\equiv X_1$.

The goal, therefore, is to solve the system of equations
$\dot x_i=\Gamma\; x_i,\ i=1,\dots,n,$
or, in a coordinate independent formulation, to obtain the
integral curves $\Phi_t:M\rightarrow M$
of $\Gamma$
\begin{equation}  \label{eom}
 (\Gamma f) (\Phi_t(x)) = \frac{d}{dt} f(\Phi_t(x)),\quad \forall f\in C^\infty (x),
\end{equation}
using quadratures
(operations of integration, elimination and partial differentiation).
The number of quadratures is given by the number of
integrals of known functions depending on a finite number of parameters
that are performed.

As we mentioned above $\Gamma$ plays a distinguished and
important r\^ole since it represents the dynamics to be integrated.
In fact,  and as the approach is
concerned with the construction of a sequence of nested Lie subalgebras $L_{\Gamma,k}$ of $L$, it will be essential that  $\Gamma$
belongs to all the subalgebras.  This construction will be carried out in several steps.

 The first one will be to reduce, by one quadrature, the original problem to a similar one but with a Lie subalgebra $L_{\Gamma,1}$ of  the Lie algebra $L$
(with $\Gamma\in L_{\Gamma,1}$)
whose elements span at every point the tangent space of the
leaves of a certain foliation.
If iterating the procedure we end up with an Abelian Lie algebra
we can, with one more quadrature, obtain the flow of the dynamical
vector field (\ref{eom}).

We determine the foliation through a family of functions
that are constant on the leaves. We first consider the ideal
$$
 L_{\Gamma,1} = \langle \Gamma\rangle + [L,L] \,,{\quad} \dim L_{\Gamma,1} = n_1,
$$
that, in order to make the notation simpler, we will assume to be generated
by the first $n_1$ vector fields of the family, i.e. $L_{\Gamma,1}=\langle \Gamma,X_2,\dots,
X_{n_1}\rangle$. This can be always achieved by choosing appropriately the
basis of $L$.

Now take $\zeta_1\in L_{\Gamma,1}^0$,
where $L_{\Gamma,1}^0$ is the annihilator of $L_{\Gamma,1}$,
i.e. the set of elements in $L^*$ that kill all vectors of $L_{\Gamma,1}$.
Now we define the 1-form $\alpha_{\zeta_1}$  {\janusz  on $M$} by its action on
the vector fields in $L$ in the following way
$$\alpha_{\zeta_1}(X)=\zeta_1(X),\quad\mathrm{for}\ X\in L.$$
As $\alpha_{\zeta_1}(X)$ is a constant function on $M$, for any vector
field in $L$, we have
$$d\alpha_{\zeta_1}(X,Y)=\alpha_{\zeta_1}([X,Y])=\zeta_1([X,Y])=0,\quad\mathrm{for}\  X,Y\in L,\ \zeta_1\in L_{\Gamma,1}^0.$$
Therefore the 1-form $\alpha_{\zeta_1}$ is closed and by application of the result of the lemma \ref{l1}
the system of partial differential equations
\begin{equation}\label{system2}
 X_i Q_{\zeta_1} =\alpha_{\zeta_1}(X_i),\quad i=1,\dots,n, \quad Q_{\zeta_1}\in C^\infty(M),
\end{equation}
has a unique (up to the addition of a constant) local solution
which  can be obtained by one quadrature.

For further purposes it will be convenient that the solution
of (\ref{system2}) depends linearly on $\zeta_1$. This can be always
achieved if we follow the construction of the lemma and we fix the same
reference point $x_0$ for any $\zeta_1$. In fact,
$\alpha_{\zeta_1}$ depends linearly on $\zeta_1$ and, if $\gamma_x$ is independent of $\zeta_1$, we have that the correspondence
$$L_{\Gamma,1}^0\ni\zeta_1\mapsto Q_{\zeta_1}\in C^\infty(M),$$
defines an injective linear map.

The previous system of equations expresses the fact that the vector fields
in $L_{\Gamma,1}$ (including $\Gamma$) are tangent to
$$
  N_1^{[Y_1]}=\{x\mid Q_{\zeta_1}(x)=\zeta_1(Y_1),\,\zeta_1\in L_{\Gamma,1}^0\}\subset M
$$
for any $[Y_1]\in L/L_{\Gamma,1}$. Locally, for an open
neigbourhood $U$, the $N_1^{[Y_1]}$'s define a smooth foliation of
$n_1$-dimensional leaves.

Now, we repeat the previous procedure by taking $L_{\Gamma,1}$ as the  Lie algebra and
any leaf $N_1^{[Y_1]}$ as the manifold.
The new subalgebra $L_{\Gamma,2}\subset L_{\Gamma,1}$ is defined by
$$
 L_{\Gamma,2} = \langle \Gamma\rangle + [L_{\Gamma,1},L_{\Gamma,1}] \,,{\quad} \dim L_{\Gamma,2} = n_2\,,
$$
and taking $\zeta_2\in L_{\Gamma,2}^0\subset L_{\Gamma,1}^*$ (the annihilator of $L_{\Gamma,2}$),
we arrive at a new system of partial differential equations
$$
 X_i Q_{\zeta_2}^{[Y_1]} =\zeta_2(X_i),\quad i=1,\dots,n_1, \quad
 Q_{\zeta_2}^{[Y_1]}\in C^\infty(N_1^{[Y_1]}) \,,
$$
that can be solved with one quadrature in such a way that  $Q_{\zeta_2}^{[Y_1]}$
depends linearly on $\zeta_2$.

For later purposes, it will be useful to extend
$Q_{\zeta_2}^{[Y_1]}$ to $U$. In order to do that we
first introduce the map
$$U\ni x\mapsto [Y_1^{^x}]\in L_{\Gamma,0}/L_{\Gamma,1}\,,$$
where $x$ and $[Y_1^{^x}]$ are  related by
the equation $Q_{\zeta_1}(x)=\zeta_1(Y_1^{^x})$, that correctly
determines the map.
Now, we define $Q_{\zeta_2}\in C^\infty(U)$ by
$$Q_{\zeta_2}(x)= Q_{\zeta_2}^{[Y_1^{^x}]}(x).$$
Note that, by construction, $x\in N_1^{[Y^{^x}_1]}$
and therefore the definition makes sense.
It is also clear that  the resulting function
$Q_{\zeta_2}(x)$
is smooth provided the reference point of the lemma changes
smoothly from leave to leave; a property that can be always
fulfilled.

The construction is then iterated by defining
$$N_2^{[Y_1][Y_2]}=\{x\mid Q_{\zeta_1}(x)=\zeta_1(Y_1), \quad Q_{\zeta_2}(x)=\zeta_2(Y_2),\
{\rm with}\
\zeta_1\in L_{\Gamma,1}^0, \zeta_2\in L_{\Gamma,2}^0\}\subset M,$$
for $[Y_1]\in L_{\Gamma,0}/L_{\Gamma,1}$ and $[Y_2]\in L_{\Gamma,1}/L_{\Gamma,2}$.
Note that $L_{\Gamma,2}$ generates at every point the tangent space
of $N_2^{[Y_1][Y_2]}$, therefore we can proceed as before.

The algorithm ends if after some steps, say $k$, the Lie algebra
$L_{\Gamma,k}=\langle X_1,\dots,X_{n_k}\rangle$,
whose vector fields are tangent to the $n_k$-dimensional
leaf $N_k^{[Y_1],\dots,[Y_k]}$, is Abelian.
In this moment the system of equations
$$
  X_i Q_{\zeta_k}^{[Y_1],\dots,[Y_k]}=\zeta_k(X_i),\quad i=1,\dots,n_{k-1},\quad
Q_{\zeta_k}^{[Y_1],\dots,[Y_k]}\in C^\infty(N_k^{[Y_1],\dots,[Y_k]}),
$$
can be solved locally by one more quadrature for any  $\zeta_k\in L_{\Gamma,k}^*$
(we note that, as the final Lie algebra $L_{\Gamma,k}$ is Abelian, the integrability condition is always satisfied
and we can take $\zeta_k$ in the whole of $L_{\Gamma,k}^*$ instead of {\janusz $L_{\Gamma,k+1}^0$)}. Then, as before, we extend
the solutions to $U$ and call them $Q_{\zeta_k}$.

With all these ingredients we can find the flow of $\Gamma$ by performing
only algebraic operations. In fact, consider the formal direct sum
$$\Xi=L_{\Gamma,1}^0\oplus L_{\Gamma,2}^0\oplus\cdots\oplus L_{\Gamma,{k}}^0\oplus L_{\Gamma,k}^*$$
that, as one can check, has dimension $n$. The linear
maps $L_{\Gamma,i}^0\ni\zeta_i\mapsto Q_{\zeta_i}\in C^\infty(U)$
can be extended to $\Xi$ so that to any
$\xi\in\Xi$ we assign a
$Q_\xi\in C^\infty(U)$.
Now consider a basis
$$\{\xi_1,\dots,\xi_n\}\subset\Xi.$$
The associated functions $Q_{\xi_j},j=1,\dots,n$ are functionally independent and satisfy
\begin{equation}\label{rectifying}
 \Gamma Q_{\xi_j}(x)= \xi_j(\Gamma) \,,{\quad}  j=1,2,\dots,n,
\end{equation}
where it should be noticed that, as $\Gamma\in L_{\Gamma,l}$ for any
$l=0,\dots,k$, the right hand side is well defined.
From (\ref{rectifying}) we see that,
in the coordinates given by $Q_{\xi_j}(x),\ j=1,\dots,n$,
the vector field $\Gamma$ has constant components and, then,
it is trivially integrated
 $$Q_{\xi_j}(\Phi_t(x))=Q_{\xi_j}(x)+ \xi_j(\Gamma) t.$$
Now, with algebraic operations, one can derive the flow $\Phi_t(x)$.
Altogether we have performed $k+1$ quadratures.

\section{Algebraic properties}

The previous procedure works if it reaches an end point, i.e.
if there is a smallest non negative integer $k$ such that
$$
 L_{\Gamma,{k}}=\langle \Gamma\rangle+[L_{\Gamma,{k-1}},L_{\Gamma,{k-1}}]\,\ {\rm for}\ k>0 \,,{\qquad} L_{\Gamma,0}=L,
$$
is an Abelian algebra.  In that case we will say that $(M,L,\Gamma)$ is
{\janusz \emph{Lie integrable of order $k+1$}}.

The content of the previous section can, thus, be summarized in the following.
\begin{theorem}
If $(M,L,\Gamma)$ is Lie integrable of order $r$,
then  the integral curves of $\Gamma$ can be obtained by $r$ quadratures.
\end{theorem}

We will discuss below some necessary and sufficient conditions for
the Lie integrability.

\begin{proposition}\label{p4}
If $(M,L,\Gamma)$ is Lie integrable, then $L$ is solvable.
\end{proposition}

\begin{proof} Denote by $L_{(i)}$ the elements of the derived series,
$L_{(i+1)}=[L_{(i)},L_{(i)}]$, $L_{(0)}=L$, (note that $L_{(i)}=L_{0,i}$).
We will show by induction that
\begin{equation}
  L_{(i)}\subset L_{\Gamma,i}.
\end{equation}
Clearly, this is true for $i=0$ and, assuming that it also holds for
some $i$, we have the following
$$L_{(i+1)}=[L_{(i)},L_{(i)}]\subset [L_{\Gamma,i},L_{\Gamma,i}]\subset L_{\Gamma,{i+1}}\,,$$
that completes the induction.

Then, if the system is Lie integrable, i.e.  $L_{\Gamma,k}$ is Abelian
for some $k$, then we have $L_{(k+1)}=0$ and, therefore, $L$ is solvable.
\end{proof}

\begin{proposition}\label{p5}
If $L$ is solvable and $A$ is an Abelian ideal of $L$,
then $(M,L,\Gamma)$ is Lie integrable for any $\Gamma\in A$.
\end{proposition}

\begin{proof}
Using that
$A$ is an ideal containing $\Gamma$, we can show that
$$A+L_{\Gamma,i}=A+L_{(i)}.$$
We proceed again by induction; if the previous holds, then
\begin{eqnarray}
A+L_{\Gamma,i+1}&=&A+[L_{\Gamma,i},L_{\Gamma,i}]=A+[A+L_{\Gamma,i},A+L_{\Gamma,i}]=\cr
&=&A+[A+L_{(i)},A+L_{(i)}]=A+L_{(i+1)}.
\end{eqnarray}
Now, $L$ is solvable if some $L_{(k)}=0$ and therefore $L_k\subset A$,
i.e. it is Abelian and, henceforth, the system is Lie integrable.
\end{proof}

Note that the particular case
in which $A=\langle \Gamma\rangle$
corresponds to the standard Lie theorem (Thm. {\ref{Lie-thm} of section 2).

Nilpotent algebras of vector fields \cite{MK88,JG90} also play an interesting role in the
integrability of vector fields.
\begin{proposition}
If $L$ is nilpotent, $(M,L,\Gamma)$ is Lie integrable for any $\Gamma\in L$.
\end{proposition}

\begin{proof} Let us consider now the central series
$L^{(i+1)}=[L,L^{(i)}]$ with $L^{(0)}=L$.
$L$ nilpotent means that there is a $k$ such that $L^{(k)}=0$.
Now, by induction, it is easy to see that
$L_{\Gamma,i}\subset \langle \Gamma\rangle +L^{(i)}$
and therefore $L_{\Gamma,k}= \langle \Gamma\rangle$.
Then, $L_{\Gamma,k}$ is Abelian and the system is Lie integrable.
\end{proof}

\section{An interesting example}

 We now analyze the particular case of a superintegrable system studied in \cite{CCR13}. In this case the system is Hamiltonian, that is, the dynamical vector field $\Gamma_H$ is obtained from a Hamitonian function
 $H$ by making use of a sympletic structure $\omega_0$ defined in the cotangent bundle $T^*Q$ ($Q$ is the configuration space). Nevertheless we are now interested in considering this system just as a dynamical
 system (without mentioning the existence of a sympletic structure) and focusing our attention on the Lie algebra structure of the symmetries.

The dynamics is given by the vector field $X_1=\Gamma$, defined in $M=\mathbb{R}^2\times\mathbb{R}^2$ with coordinates $(x,y,p_x,p_y)$, by
$$
  \Gamma = p_x\pd{}{x} + p_y\pd{}y  - \frac{k_2}{y^{ 2/3}}\pd{}{p_x}
  + \frac{2}{3} \frac{k_2\,x+k_3}{y^{ 5/3}} \pd{}{p_y} \,,
$$
where $k_2$ and $k_3$ are arbitrary constants.
Now, with $X_i$, $i=2,3,4$, we denote the vector fields
$$
\begin{array}{rcl}
X_2&=& {\displaystyle
\left(6\, p_x^2+3\, p_y^2+k_2\frac {6x}{y^{ 2/3}} + k_3 \frac {6}{y^{ 2/3}}\right)\pd{}{x}+(6\, p_xp_y+9\, k_2y^{ 1/3})\pd{}y}  \\
&-&\displaystyle{k_2\frac {6}{y^{ 2/3}}\, p_x
\pd{}{p_x}+\left(4k_2\frac {x}{y^{ 5/3}}-3\frac {1}{y^{ 2/3}}\,p_y\right)\pd{}{p_y}}   \,,
\end{array}
$$
$$
\begin{array}{rcl}
X_3&=&{\displaystyle\left(4\, p_x^3+4\, p_xp_y^2+\frac{8(k_2x+k_3)}{y^{ 2/3}}p_x+12k_2\, y^{ 1/3}\,p_y\right)
\pd{}{x}}\\
&+&\left(4p_x^2\, p_y+12k_2\, y^{ 1/3}\,p_x\right)\displaystyle{\pd{}{y}}
- {4k_2\frac {1}{y^{ 2/3}}p_x^2\,\pd{}{p_x}}  \\
&+&\left(\displaystyle{ \frac{8}{3} \frac{k_2x+k_3}{y^{5/3}} p_x^2}
- 4k_2 \frac{1}{y^{ 2/3}} p_xp_y - 12\,k_2^2 \frac{1}{y^{1/3}}\right)
\displaystyle{\pd{}{p_y} }  \,,
\end{array}
$$
and
$$\begin{array}{rcl}
X_4&=&{\displaystyle\left(6p_x^5+12\, p_x^3p_y^2+24 \frac{k_3 + k_2 x}{y^{2/3}}p_x^3+108\,k_2 y^{1/3} p_x^2 p_y  +324\, k_2^2   y^{2/3} p_x\right)\pd{}{x}}\\&+&
{\displaystyle \left(6\,  p_x^4 p_y+36\,k_2 y^{1/3} p_x^3   \right)\pd{}y-6\, \left(\frac{ k_2 }{y^{2/3}} p_x^4-972 k_2^3 \right)\pd{}{p_x}}
\cr&+&{\displaystyle \left(4\,\frac{k_3 + k_2 x}{y^{5/3}}p_x^4-12\frac{ k_2}{ y^{2/3}}-108\, k_2^2\frac 1{y^{1/3}}p_x^2 \right)\pd{}{p_y}} \,.
\end{array}
$$
Then, we have
\begin{itemize}
\item [(i)] The three vector fields $X_i$ Lie commute with $X_1=\Gamma$
$$
[\Gamma , X_i]=0 \,,{\quad} i=2,3,4.
$$
\item [(ii)] The Lie brackets of the $X_i$ between themselves are given by
$$
 [X_2, X_3]=0 \,,{\qquad}
 [X_2, X_4] = 1944\, k_2^3\,\Gamma\,,{\qquad}
 [X_3, X_4] = 432\, k_2^3 \,X_2  \,.
$$
\end{itemize}

Therefore, we have the following properties. First, $\Gamma$ and the three
vector fields $X_2,X_3,X_4$ generate a four-dimensional real Lie algebra ${L}$.
Second,  the derived algebra ${L}_{(1)}\subset {L}$  is two-dimensional and Abelian because
it is generated by $\Gamma$ and $X_2$. Finally, the second derived
algebra ${L}_{(2)}$ reduces to the trivial algebra, that is,
${L}_{(2)}=[{L}_{(1)},{L}_{(1)} ]=\{0\}$. Therefore the Lie algebra ${L}$ is
solvable. However,
${L}^{(2)}=[{L},{L}_{(1)} ]$ is not trivial but ${L}^{(2)} $ is the one-dimensional  ideal in ${L}$ generated by $\Gamma$,
and this implies that the Lie algebra is nilpotent. Consequently $(M,L,X_i ) $ is integrable for any index $ i$.

{\janusz \section{Distributional integrability}}

The previous construction is in some sense too rigid or
too restrictive.
For instance, the very simple system in $\CR^n$
with dynamical vector field
$$
 \Gamma=f(x)\partial_1
$$
which corresponds to the system of equations
$$
\dot x^1= f(x),\quad  \dot x^2=0,\quad  \dots,\quad  \dot x^n=0,
$$
can be easily solved by quadratures.
However, if one considers the natural choice
$$L=\langle \Gamma,\partial_2,\dots,\partial_n\rangle\,,$$
the vector fields do not close a Lie algebra over the real numbers.
It would be worth extending the results in the previous sections
to allow for Lie algebras over the ring of functions and accommodate this and other
interesting cases. In the following we will pursue this goal.

In order to proceed, we shall need some preliminary definitions
and results.

\begin{definition}\hfill\break
For any subset $S\subset{\FX}(M)$, we denote by $\D_S$ the
{\janusz  $C^\infty(M)$-module generated by $S$, i.e.
$$
 \D_S=\left\{\sum_i f^i X_i\in {\FX}(M) \mid f^i\in C^\infty(M),\, X_i\in S\right\}.
$$
As $\D_S$ is the module of sections of the corresponding generalized distribution, we will
also refer to $\D_S$ as to a distribution.}
\end{definition}

\begin{definition}\hfill\break
We say that a real vector space $V\subset{\FX}(M)$ is \emph{regular}
if  $V$ is isomorphic to its restriction
$V_p\subset T_pM$ at any point $p\in M$, and \emph{completely regular} if it is regular and
$V_p=T_pM$.
\end{definition}
The previous definitions immediately imply the following.
\begin{proposition}\hfill\break
1) Any subspace of a regular space is regular.\hfill\break
2) For any  two subspaces $W_1,W_2$
of a regular space we have $\D_{W_1}\cap \D_{W_2}=\D_{W_1\cap W_2}$.
\end{proposition}
Based on these properties we introduce the following definition:
\begin{definition}
Given a completely regular space $V\subset{\FX}(M)$ and a subset $S\subset\FX(M)$, we shall call the \emph{core of $S$} in $V$, denoted by $\core{S}$, the smallest subspace of
$V$ such that $S\subset \D_{\core{S}}$.
\end{definition}
That the previous definition makes sense and the core of any set
$S\subset \D_V$ exists, is contained in the following.
\begin{proposition} For a completely regular space $V\subset {\FX}(M)$ and any $S\subset\FX(M)$,
$$
\core{S}
=\bigcap_{W\in{\cal W}} \langle W\rangle,
\quad {\rm with}\quad {\cal W}= \{ W\subset V \mid S\subset \D_W\}\,.
$$
\end{proposition}
Usually $V$ will be fixed once for all and this is the reason
why it does not appear in the notation.

Now, for a completely regular space $V\in {\FX}(M)$
and a dynamical vector field $\Gamma\in V$, we introduce
the following series:
$V_{\Gamma,0}=V$ and
$$
 V_{\Gamma,m}=\langle \Gamma\rangle + \core{[V_{\Gamma,m-1},V_{\Gamma,m-1}]} \,.
$$
Observe that $V_{\Gamma,m}\subset V_{\Gamma,m-1}$, which is easily shown by
induction. If there exists a $V_{\Gamma,k}$ for non negative $k$
which is the first Abelian
subspace in the series, we shall say that $(M,V,\Gamma)$ is
\emph{distributionally  integrable of order $k+1$}.
Then, we can state the main result of this section.
\begin{theorem}
If $(M,V,\Gamma)$ is distributionally integrable of order $r$,
then the  vector field $\Gamma$ can be integrated by $r$ quadratures.
\end{theorem}
\begin{proof}
The procedure to obtain the solution of the system of differential
equations is the same as the one sketched in the first section.
For any $\zeta_1\in V^*$ annihilating $V_{\Gamma,1}$,
we introduce the 1-form $\alpha_{\zeta_1}$ such that
$\alpha_{\zeta_1}(X_i)=\zeta_1(X_i)$. Then, one immediately sees
that $d\alpha_{\zeta_1}(X_i,X_j)=0$ and, therefore, the 1-form is closed.
Then, by applying Lemma \ref{l1}, the system of partial differential equations
$$X_i Q_{\zeta_1} =\zeta_1(X_i),\quad i=1,\dots,n,\quad Q_{\zeta_1}\in C^\infty(U),$$
has a solution obtained with one quadrature.

By construction, vector fields in $V_{\Gamma,1}$ are tangent to the level set
of the $Q_{\zeta_1}$'s and, then, we can reduce the problem to that
submanifold.

Iterating the procedure we finally solve the problem by performing
$r$ quadratures that ends the proof.
\end{proof}

To illustrate this construction we will provide two examples
that represent, somehow, two opposite ends.

\vskip 3mm
\begin{example}
The first example was mentioned at the beginning of this section.
Consider $M=\CR^n$, $\Gamma=f(x)\partial_1$ with $f$ being a nowhere vanishing
smooth function,
and
$$V=\langle \Gamma,\partial_2,\dots,\partial_n\rangle\,.$$
Then, we immediately see that
$[\Gamma,\partial_i]\in  \D_{\langle \Gamma\rangle}$ for any $i$
and therefore $V_1=\langle \Gamma\rangle$,
so the system of equations is solved with $2$ quadratures.
\end{example}
\begin{example}
The second example is at the opposite end as it requires $n$ quadratures.
In this case we take
\begin{eqnarray}
\Gamma=f(x)\big(\hskip -.65cm&&
\partial_1+ g^2(x^1)\partial_{2}+\dots+
g^{n-1}(x^1,\dots,x^{n-2})\partial_{{n-1}}+
\cr&&+g^{n}(x^1,\dots,x^{n-1})\partial_{n} \big),
\end{eqnarray}
with $f(x)\not=0$ everywhere and
$$V=\langle \Gamma,\partial_{2},\dots,\partial_{n}\rangle.$$
It is immediate to show that
$V_{\Gamma,1}=\langle \Gamma,\partial_{3},\dots,\partial_{n}\rangle$,
$V_{\Gamma,2}=\langle \Gamma,\partial_{4},\dots,\partial_{n}\rangle$,
and finally $V_{\Gamma,n-1}=\langle \Gamma\rangle.$
This shows that the system is distributionally integrable
and requires $n$ quadratures for its solution.
\end{example}
Note that in the previous examples there is an arbitrary, nowhere vanishing
function $f$ that multiplies the dynamical vector field.
This is, actually, the general situation as it is
stated in the following proposition.
\begin{proposition}
Suppose that  $(M,V,\Gamma)$, with
$V=\langle \Gamma,X_2,\dots,X_n\rangle$,
is distributionally integrable of order $r$.
Then, for any nowhere-vanishing $f\in C^\infty(M)$, 
 the system $(M,V',f\Gamma)$ with
$V'=\langle f\Gamma,X_2,\dots,X_n\rangle$
is distributionally integrable of order $r'\in\{r-1, r, r+1\}$
\end{proposition}
\begin{proof} By induction, it is easy to see that if $\{\Gamma,Y_2,\dots,Y_l\}\subset V_{\Gamma,m}$
is a basis of $V_{\Gamma,m}$, then $\{f\Gamma,Y_2,\dots,Y_l\}$ forms a basis
of $V_{f\Gamma,m}$. Therefore, if $V_{\Gamma,r-1}$ is Abelian, then $V_{f\Gamma,r}$ is Abelian too,
i.e.  $r'\leq r+1$. But the relation is obviously symmetric, therefore
we must have $r\leq r'+1$, from which we get $r-1\leq r'\leq r+1$.
\end{proof}

We can also translate the properties of the previous section
to this generalized setup.
We shall call a completely regular $V$ \emph{distributionally solvable} if
the series $V_{(i)}
=\core{[V_{(i-1)},V_{(i-1)}]}$, with $V_{(0)}=V$, stabilizes trivially, $V_{(n)}=\{ 0\}$.
Denote with $\D_{(s)}$ the distribution $\D_{V_{(s)}}$. It is clear that $\D_{(s)}\subset \D_{(s')}$ for $s>s'$ and $[\D_{(s)},\D_{(s)}]\subset \D_{(s)}$, so that these distributions are involutive, hence integrable.

We will say analogously that a completely regular $V$ is
\emph{distributionally nilpotent} if
the generalized central series:
$$V^{(i)}=\core{[V^{(i-1)},V]},\quad {\rm with}\quad V^{(0)}=V,$$
stabilizes at $\{0\}$.
\begin{example}\label{e1}
Consider in $\FX(\mathbb R^n)$ the vector subspace $V$ spanned by a basis $X_i$, $i=1,\dots,n$, of vector fields of the following `triangular form':
$$X_i=\partial_i+\sum_{k>i}f_i^k(x)\partial_k\,.$$
Then, $V$ is distributionally solvable. Indeed, the vector fields in $[V,V]$ have no $\partial_1$ components, so {\janusz  $V_{(1)}$} is spanned by $X_2,\dots,X_n$. Inductively, {\janusz  $V_{(n-1)}$} is spanned by
$X_n=\partial_n$ and {\janusz $V_{(n)}=\{ 0\}$}.
\end{example}

\begin{example}\label{e2}
Consider in $\FX(\mathbb R^n)$ the vector subspace $V$ spanned by a basis $X_i$, $i=1,\dots,n$, of vector fields of the following `strong triangular form':
$$X_i=\partial_i+\sum_{k>i}f_i^k(x^1,\dots,x^{k-1})\partial_k\,,$$
where the coefficients $f_i^k$ depend on variables $x^{1},\dots,x^{k-1}$ only.
Then, $V$ is distributionally nilpotent.
Indeed, as before {\janusz $V^{(1)}$} is spanned by $X_2,\dots,X_n$ and, inductively, {\janusz $[V,V^{(s)}]$} is spanned by $X_{s+1},\dots,X_n$, so  {\janusz $V^{(n)}=\{ 0\}$}.
\end{example}

The proof of the following proposition is completely analogous to the proofs of Proposition \ref{p4} and Proposition \ref{p5}.
\begin{proposition}\
\begin{description}
\item{(a)} If $(M,V,\Gamma)$ is distributionally integrable,
then $V$ is distributionally solvable.
\item{(b)} If $V$ is distributionally nilpotent, then
$(M,V,\Gamma)$ is distributionally integrable for any $\Gamma\in V$.
\end{description}
\end{proposition}

We end up with presenting  descriptions of distributionally solvable or nilpotent  $V\subset\FX(M)$ showing that the Examples \ref{e1} and \ref{e2} are in a sense universal {\janusz  (compare with \cite{MK88,JG90})}.

{\janusz  Let $V\subset\FX(M)$ be distributionally solvable, $r$ be the smallest natural number such that
$V_{(r)}=\{ 0\}$, and $d_s$, $s=1,\dots, r$, be the dimension of the space $V_{(s-1)}/V_{(s)}$. } Put
$w_s=d_1+\cdots +d_{s}$, for $s\geq 1$, {\janusz  to be the dimension of $V/V_{(s)}$}.
\begin{theorem} For any $p\in M$, there is a basis $\{X_1,\dots,X_n\}$ of $V$  and there exist local coordinates {\janusz $(x^u)$} around $p$ such that, for $w_{s-1}<i\le w_s $, $s=1,\dots, r$, {\janusz  the vector field $X_i$ is of the form}
$$X_i=\partial_i+\sum_{k>w_s}f_i^k(x)\partial_k\,.$$
\end{theorem}
\begin{proof}
Choose $X_1,\dots,X_{d_1}$ representing a basis of {\janusz $V/V_{(1)}$}. Let $\alpha^1,\dots,\alpha^{d_1}$ be 1-forms on $M$ such that $\alpha^j(X_i)=\delta^j_i$ and {\janusz $\alpha^j(V_{(1)})=0$}, for $i,j=1,\dots,w_1$. Exactly as we have seen earlier, the 1-forms are closed and they have potentials $x^i$ around $p$ of the form
$$
  x^i(q)=\int_{\gamma_q}\alpha^i,
$$
where $\gamma_q$ is a smooth path joining $p$ with $q\in M$.

Now we can choose {\janusz $X_{d_1+1},\dots,X_{w_2}\in V_{(1)}$} representing a basis of {\janusz $V_{(1)}/V_{(2)}$}. As the distribution generated by {\janusz $V_{(1)}$} is clearly involutive, it defines a foliation by the level sets $M(a)$ of $(x^1,\dots,x^{d_1})$, so we can choose a smooth submanifold $N_1$ of dimension $d_1$ through $p$ intersecting the local leaf $M(a)$  of {\janusz $V_{(1)}$}, where $a\in\mathbb R^{d_1}$ is sufficiently close to 0, at a single point $p(a)$.

Inductively, for any such $a$, we have on $M(a)$ the 1-forms $\alpha^{d_1+1}_a,\cdots, \alpha^{w_2}_a$ such that $\alpha^j(X_i)=\delta^j_i$ and {\janusz $\alpha^j(V_{(2)})=0$}, for $i,j=d_1+1,\dots,w_2$. The 1-forms are, this time leaf-wise, closed and they have potentials $x^i$ around $p$, $i=d_1+1,\cdots, w_2$, of the form
\begin{equation}\label{solution2}
  x^i(q)=\int_{\gamma_q(a)}\alpha^i,
\end{equation}
where $\gamma_q(a)$ is a smooth path lying entirely  inside the leaf $M(a)$ and joining $q\in M(a)$ with the point $p(a)$. The functions $x^i$, with $i=d_1+1,\cdots, w_2$, are globally defined and smooth, and they are, this time only leaf-wise, potentials for $\alpha^i$. In any case,
$X_i(x^j)=\delta^j_i$ and {\janusz $V_{(2)}(x^j)=0$}, for $i,j=d_1+1,\dots,w_2$.

Proceeding in this way inductively, we prove the theorem.
\end{proof}

If we assume that $V\subset\FX(M)$ is distributionally nilpotent and define $R$, $D_s$, $W_s$
like $r$, $d_s$, and $w_s$ above, but using the sequence of subspaces {\janusz $V^{(s)}$ instead of $V_{(s)}$}, we get a basis $(X_i)$ of $V$ and local coordinates $x^i$ as before, with
$$X_i=\partial_i+\sum_{k>W_s}f_i^k(x)\partial_k\,,$$
for any $i$ satisfying $W_{s-1}<i\le W_s $, $s=1,\dots, R$.
This time, however, {\janusz $[X_i,V^{(s')}]$} lies in the distribution spanned by {\janusz $V^{(s'+1)}$}.
This means, for $s'>s$, that $\partial_j(f_i^k)=0$ for $k\le W_{s'}$ {\janusz  and $j>W_{s'-1}$}, thus we get a stronger triangular form for elements of $V$ as follows.
\begin{theorem}{\janusz  Let $V\subset\FX(M)$ be distributionally nilpotent, $R$ be the smallest natural number such that
$V^{(R)}=\{ 0\}$, and
$W_s$ be the dimension of $V/V^{(s)}$. Then,
for any $p\in M$, there is a basis $X_1,\dots,X_n$ of $V$  and local coordinates $(x^u)$ around $p$ such that, for $W_{s-1}<i\le W_s $, $s=1,\dots, R$, {\janusz  the vector field $X_i$ is of the form}
$$X_i=\partial_i+\sum_{k>W_s}f_i^k(x)\partial_k\,,$$
where all coefficients $f_i^k$, with $k\le W_{s'+1}$, $s\le s'$, depend on variables $x^1,\dots, x^{W_{s'}}$ only.}
\end{theorem}

\section*{Acknowledgments}

This work was partially  supported by the research projects MTM--2012--33575, FPA--2012--35453 (MINECO) and DGA E24/1, E24/2 (DGA, Zaragoza). Research of
 J.~Grabowski was founded by the Polish National Science Centre grant HARMONIA under the contract number DEC-2012/04/M/ST1/00523.


\end{document}